\documentclass[runningheads, a4paper]{llncs}
\usepackage[affil-it]{authblk}
\usepackage[english]{babel}
\usepackage{amsmath}
\usepackage{amssymb}
\usepackage{mathrsfs}
\usepackage{delarray}
\usepackage{tabularx}
\usepackage[noend]{algpseudocode}
\usepackage{environ}
\usepackage{authblk}
\usepackage{stmaryrd}
\usepackage{calc}
\usepackage{algorithm}
\usepackage{algpseudocode}

\usepackage{graphicx}
\graphicspath{ {./images/} }

\usepackage{xcolor}

\usepackage{tikz}
\usetikzlibrary{arrows, automata, positioning, arrows.meta, calc}

\tikzset{configuration/.style = {state, rectangle, minimum height=0.5cm},}

\let\emptyset\varnothing

\newcommand{\B}{\mathbb{B}}
\newcommand{\N}{\mathbb{N}}

\newcommand{\xor}{\oplus}
\newcommand{\bigxor}{\bigoplus}

\algblock{Input}{EndInput}
\algnotext{EndInput}
\algblock{Output}{EndOutput}
\algnotext{EndOutput}

\title{Fast solutions to $k$-parity and $k$-synchronisation using parallel automata networks}
\date{}

\author{Pac\^ome Perrotin\inst{1}\and 
Eurico Ruivo\inst{1,}\inst{2}\and 
Pedro Paulo Balbi\inst{1,}\inst{2} 
}
\institute{
P\'{o}s-Gradua\c{c}\~{a}o em Engenharia El\'{e}trica e Computa\c{c}\~{a}o
\and
Faculdade de Computa\c{c}\~{a}o e Inform\'{a}tica\\
Universidade Presbiteriana Mackenzie\\
Rua da Consola\c{c}\~{a}o 896, Consola\c{c}\~{a}o; 01302-907 S\~{a}o Paulo, SP -- Brazil
\email{pacome.perrotin@gmail.com, eurico.ruivo@mackenzie.br, pedrob@mackenzie.br}
}

\begin{document}

\maketitle

\begin{abstract}
  We present a family of automata networks that solve the $k$-parity problem when run in parallel.
  These solutions are constructed by connecting cliques in
  a non-cyclical fashion.
  The size of the local neighbourhood is linear in the
  size of the alphabet, and the convergence time is proven to always be the
  diameter of the interaction graph. We show that this family of solutions can
  be slightly altered to obtain an equivalent family of solutions to the
  $k$-synchronisation problem, which means that these solutions converge from
  any initial configuration to the cycle which contains all the uniform
  configurations over the alphabet, in order.
\end{abstract}

\section{Introduction}

Automata networks are general, distributed and discrete models of computation
that are studied for their capacity for complex behaviour despite simple local rules.
They are widely used as models of gene regulatory networks~\cite{Kauffman1969,Thomas1973,Bernot2003}.
Cellular automata are a related family of discrete models with uniform geometry and local rule.
with wide applications for the simulation of real-world phenomena~\cite{CA31}. Examples
inclue disease spread~\cite{CA39}, urban growth~\cite{CA2} or fluid dynamics~\cite{CA54}.

The parity problem is a classical benchmark
problem in artificial intelligence, dating
back to~\cite{balbi10}, where it referred to
the challenge of computing the parity
of a binary sequence without resorting to
scanning the entire sequence.
It has been adapted to a problem over
cellular automata~\cite{balbi11,balbi15,balbi17},
in which solutions have been found using
a specific pair and a specific triplet of 
rules in different positions for both
the parity problem and the synchronisation
problem.
Other solutions exist that allow the local
rule to change over time~\cite{balbi6,balbi7},
which were generalised to the modulo 3 case
in~\cite{balbi18} and the modulo N case
in~\cite{balbi8}, meaning that the decision is about whether the
number of 1s in the configuration is a
perfect multiple of N. Finally,
conditional solutions
using single rules have been found using
rule 60~\cite{balbi12} and rule 150~\cite{balbi14}.

In this paper we present automata network solutions
to particular distributed problems, namely the parity problem and the
synchronisation problem. The former is solved by any Boolean network that
converges,
from any starting configuration, to the uniform configuration made by the parity of the values of the initial configuration.
The latter is solved by any Boolean network that always converges to the limit
cycle containing the uniform configurations $0^n$ and $1^n$ for $n$ the size
of the network.
Both of these
problems are generalised for larger alphabets;
in this case, instead of computing
the parity of the values of the initial configuration which are in
$\{0, 1\}$, the network must converge to the sum modulo $k$ of the values of
the initial configuration which are in $\{0, \ldots, k - 1\}$. We call this
the $k$-parity problem. Similarly, the $k$-synchronisation problem requires the
solution to always converge to a limit cycle of size $k$ containing all of the
uniform configurations. While there is no canonical order in which these
configurations must proceed in the limit cycle, the solutions presented in
this paper will cycle through them in increasing order.

Section~\ref{sec-def} goes through all the background definitions needed. Section~\ref{sec-sol} presents our solution to the 
parity problem, and its equivalent in the synchronisation case, and Section~\ref{sec-lar} shows their generalisations for
alphabets of any finite size. Section~\ref{sec-con} details properties and proof
over the convergence times of the solutions, and describe how it can be
made faster at the cost of the size of local neighbourhoods.

\section{Definitions}
\label{sec-def}

\subsection{Background definitions}

Let $\Sigma$ be a finite alphabet. We call a \emph{configuration} of size $n$ any
vector $x \in \Sigma^n$. For $i \in \{1, \ldots, n\}$, we denote by $x_i$ the
value of the component of index $i$ in $x$. For $I \subseteq \{1, \ldots, n\}$,
we denote $x|_I$ the projection of $x$ over the indexes $I$, which is a vector
over $\Sigma^I$.

A \emph{directed graph} is a pair $G = (V, E)$, where $V$ is the \emph{set of vertices} and
$E \subseteq V \times V$ is the \emph{set of edges}. A set $V' \subseteq V$ is called
a \emph{clique} if and only if $V' \times V' \subset E$, that is, there is an edge
between every pair of vertex in $V'$, including from every edge to itself.

For $G =(V,E)$ a directed graph and $v \in V$, the \emph{neighbourhood of radius $r$} of $v$ is defined as the set of
vertices at distance no more than $r$ from $v$; that is, $v'$ is in the neighbourhood
of radius $r$ of $v$ if and only if there exists a path from $v$ to $v'$ or from $v'$
to $v$ that takes $r$ edges, or less.

\subsection{Automata networks}

An \emph{automata network} (AN) is a function
$F : \Sigma^n \to \Sigma^n$, where $n$ is the size of the network.
Such a network is typically divided into components, called automata;
for any $i \in \{1, \ldots, n\}$, we denote $f_i : \Sigma^n \to \Sigma$
the local function such that $F(x)_i = f_i(x)$.

\begin{example}
  \label{ex-an}
  Let $\Sigma = \{0, 1\}$, and $F$ be an automata network of size $n = 3$
  with local functions $f_1(x) = \neg x_2 \vee x_3$, $f_2(x) = x_1$ and
  $f_3(x) = x_2$.
\end{example}

We call \emph{interaction digraph} of $F$ the digraph with the automata
$\{1, \ldots, n\}$ for nodes and $(u, v)$ is an edge if and only if
there exist $x, x' \in \Sigma^n$ such that
$x_i \neq x'_i \Leftrightarrow i = u$ and $f_v(x) = f_v(x')$.
In other terms,
$x$ and $x'$ are only distinct in node $u$, and that is sufficient
to change the evaluation of $f_v$ over both configurations. It means that the value of the configuration in node $u$ has a decisive role on the evaluation of $v$.
The interaction graph of the network detailed in Example~\ref{ex-an} is
illustrated at the top of Figure~\ref{fig-examples}.

The \emph{parallel dynamics} of $F$,
also called the \emph{configuration space} of $F$, is the graph
with node set $\Sigma^n$, such that $(x, y)$ is an edge if and only if
$F(x) = y$.
A \emph{trap set} of the dynamics is a subset $T \subset \Sigma^n$ such that
$x \in T \Rightarrow F(x) \in T$. We call the minimal trap sets of $F$
the \emph{attractors} of $F$, and the subgraph of the dynamics containing the
attractors is called the \emph{limit dynamics} of $F$.
In the parallel dynamics case, the limit dynamics are a collection of cycles,
called limit cycles. A limit cycle is said to have \emph{size} $c$ if $c$ is the minimal integer such that $F^c(x)=x$ for all $x$ in the limit cycle. In particular, limit cycles of size $1$
are called the \emph{fixed points} of the network.

Given an automata network $F:\Sigma^n\rightarrow \Sigma^n$, $x\in\Sigma^n$ is said to converge to a limit cycle $C$ if there is an integer $t_0$ such that, for any integer $t\geq t_0$, $F^t(x)\in C$.

\begin{example}
  The dynamics of the automata network detailed in Example~\ref{ex-an} are
  illustrated as part of Figure~\ref{fig-examples}. The limit dynamics of
  this network are the fixed point $111$, and the limit cycle of size $3$
  composed of $110$, $011$ and $101$.
\end{example}

\subsection{Network convergence problems}

We now describe the parity, $k$-parity and synchronisation problems.

We say that $F$ \emph{decides} if its limit dynamics are restricted to
the fixed points $\{s^n \mid s \in \Sigma\}$. In the Boolean case where
$\Sigma = \B = \{0, 1\}$, these fixed points are $0^n$ and $1^n$. A deciding
Boolean network converges to either all $0$s or all $1$s from any starting
configuration. More precisely in the Boolean case, we say that $F$
\emph{decides parity} if it decides, and if for any starting
configuration $x$, $F$ converges to $1^n$ if $x$ has an odd number of $1$s,
and converges to $0^n$ otherwise. Note that this implies that $n$ is odd,
as otherwise $1^n$ would itself have an even number of $1$s, and would then
have to converge to $0^n$.

In the more general case in which $\Sigma =\{0,1,\cdots,k-1\}$, $F$ \emph{decides} $k$\emph{-parity} if it decides and any configuration $x\in\Sigma^n$ converges to $\left\{\left(s\right)^n\right\}$, where $s\in\Sigma$ is the $k$-parity of $x$, that is, the sum of the
values of the initial configuration modulo $k$ is $s \in \Sigma$, then convergence should be to $s^n$. In that case, $n$ cannot be a multiple of $k$, otherwise the problem is not well-defined.

We say that $F$ \emph{synchronises} if its limit dynamics are restricted to
one limit cycle of size $|\Sigma|$ containing all the configurations
in $\{s^n \mid s \in \Sigma\}$. For example, a Boolean synchronising network
starting from any configuration will converge to either $0^n$ or $1^n$,
and then swap between them at each step.

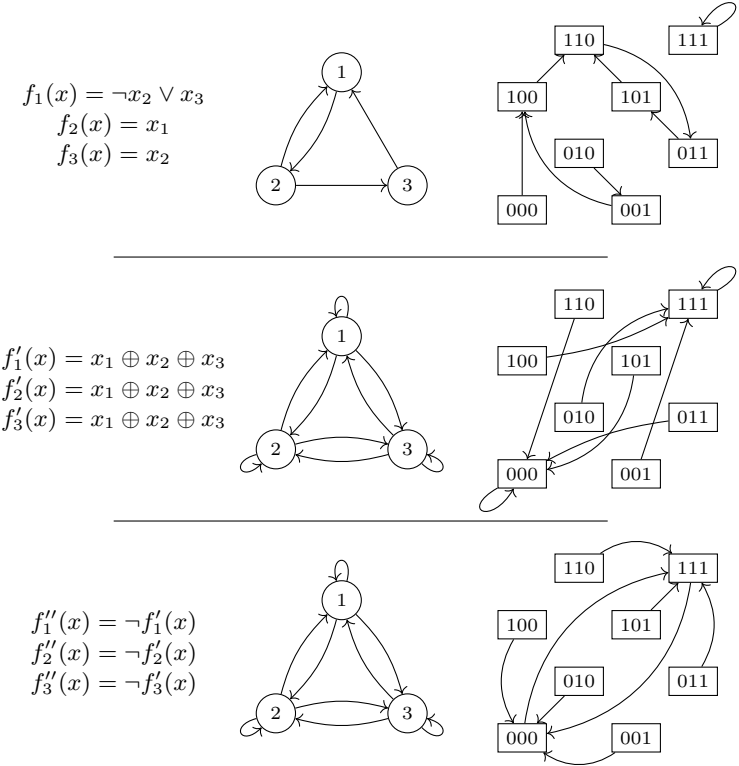
\begin{figure}[t!]
  \begin{center}
    \begin{tikzpicture} [->]


      \def \ftoiD {3}
      \def \itodD {3.5}
      \def \lineD {-3.5}

      \foreach \name\number in {A/0, B/1, C/2} {
        \node (functions\name) at (0, { 0 + \number * \lineD } ) {};
        \node (interaction\name) at (\ftoiD, { 0 + \number * \lineD } ) {};
        \node (dynamics\name) at ({\ftoiD + \itodD}, { 0 + \number * \lineD } ) {};
      }


      \foreach \a\b in {A/B, B/C} {
        \draw[-] ($.5*(functions\a)+.5*(functions\b)$) --
        ($.5*(dynamics\a)+.5*(dynamics\b)$);
      }


      \def \funSpacing {.4}

      \foreach \name\fname\fa\fb\fc in {
        A/f/\neg x_2 \vee x_3/x_1/x_2,
        B/f'/x_1\xor x_2\xor x_3/x_1\xor x_2\xor x_3/x_1\xor x_2\xor x_3,
        C/f''/\neg f'_1(x)/\neg f'_2(x)/\neg f'_3(x)
      } {
        \node (f1\name) at ($(functions\name) + (0,\funSpacing)$)
        {$\fname_1(x) = \fa$};
        \node (f2\name) at (functions\name)
        {$\fname_2(x) = \fb$};
        \node (f3\name) at ($(functions\name) - (0,\funSpacing)$)
        {$\fname_3(x) = \fc$};
      }


      \def \triangleRadius {1}
      \def \minSize {.2}

      \node (triangleOffset) at (0,-.3) {};

      \foreach \name in {A, B, C} {
        \foreach \n\x\y in {
          1/0/1,
          2/-.866/-.5,
          3/.866/-.5
        } {
          \node[state, minimum size = \minSize] (\n\name)
          at ($(interaction\name) + (triangleOffset) + \triangleRadius*(\x,\y)$)
          {\scriptsize \n};
        }
      }


      \def \bend {15}

      \foreach \name in {B, C} {
        \foreach \a\b in {1/2, 2/3, 1/3} {
          \draw
          (\a\name) edge[bend left = \bend] (\b\name)
          (\b\name) edge[bend left = \bend] (\a\name);
        }
      }

      \draw
      (1A) edge [bend left = \bend] (2A)
      (2A) edge [bend left = \bend] (1A)
      (2A) edge (3A)
      (3A) edge (1A);


      \def \wide{30}
      \def \looseness{7}

      \foreach \name in {B, C} {
        \foreach \n\angle in {1/90, 2/210, 3/330} {
          \draw (\n\name) to
          [out={\angle - \wide/2}, in={\angle + \wide/2}, looseness=\looseness]
          (\n\name);
        }
      }


      \def \dynamicsScale {1.5}

      \node (dx) at ($\dynamicsScale*(0,1)$) {};
      \node (dy) at ($\dynamicsScale*(.5,.5)$) {};
      \node (dz) at ($\dynamicsScale*(1,0)$) {};

      \node (base) at ($-.5*(dx)-.5*(dy)-.5*(dz)$) {};

      \foreach \name in {A, B, C} {
        \node (base\name) at ($(base) + (dynamics\name)$) {};
        \foreach \x in {0, 1} {
          \foreach \y in {0, 1} {
            \foreach \z in {0, 1} {
              \node[configuration, minimum size = .3] (\name\x\y\z)
              at ($(base\name) + \x*(dx)+\y*(dy)+\z*(dz)$) {\scriptsize \x\y\z};
            }
          }
        }
      }


      \foreach \a\b in {
        A000/A100, A100/A110, A010/A001, A011/A101, A101/A110,
        B001/B111, B110/B000,
        C010/C000, C101/C111}
      {
        \draw (\a) edge (\b);
      }

      \foreach \a\b in {
        A001/A100, A110/A011,
        B010/B111, B101/B000,
        C111/C000, C000/C111, C001/C000, C110/C111}
      {
        \draw (\a) edge [bend left = 35] (\b);
      }

      \foreach \a\b in {
        C100/C000, C011/C111}
      {
        \draw (\a) edge [bend right = 30] (\b);
      }

      \foreach \a\b in {B011/B000, B100/B111} {
        \draw (\a) edge [bend right = 10] (\b);
      }


      \def \wide{30}
      \def \looseness{7}

      \foreach \name\angle in {A111/45, B000/225, B111/45} {
        \draw (\name) to
        [out={\angle - \wide/2}, in={\angle + \wide/2}, looseness=\looseness]
        (\name);
      }

    \end{tikzpicture}
  \end{center}
  \caption{
    Representations of different examples of automata networks composed of their
    respective
    local functions (left), interaction graph (center)
    and parallel dynamics (right). These automata networks are, respectively, 
    the one detailed in Example~\ref{ex-an} (top), a parity deciding network
    (center), and a synchronising network (bottom).
    \label{fig-examples}
  }
\end{figure}

\section{The solutions}
\label{sec-sol}

\subsection{Their basis}
\label{subsec:basis}

We now define our solution to the $k$-parity decision problem in the Boolean case,
for any odd value of $n$, on a family of networks we refer to by \emph{Clique Trees}, or CTs. Let us define them.

Starting with $n = 1$, we define $CT_1$ as the singleton containing the
network of size $1$, with $F(x) = x$.
For any $n \in \N$, we define the elements $CT_{n + 2}$ as extensions of
the elements of $CT_n$: for $F \in CT_n$, for any $i \in \{1, \ldots, n\}$,
we define $F': \B^{n + 2} \to \B^{n + 2}$ such that:

\begin{itemize}
  \item
    $\forall j \in \{1, \ldots n\}, j \neq i
    \Rightarrow F'(x)_j = F(x|_{\{1, \ldots, n\}})_j$
  \item
    $F'(x)_i = F(x|_{\{1, \ldots, n\}})_i \xor x_{n + 1} \xor x_{n + 2}$
  \item
    $F'(x)_{n + 1} = F'(x)_{n + 2} = x_i \xor x_{n + 1} \xor x_{n + 2}$,

\end{itemize}

\noindent where $\xor$ denotes the XOR function, or sum modulo two, also known as
the parity function.
As even sized solutions are not possible, we also note that $CT_{2n} = \emptyset$ for $n \in \N$. Finally, we collect all the solutions in the set $CT$ such that

\[CT = \bigcup_{n \in \N} CT_n.\]

\begin{example}
  \label{ex-sol11}
  Let $F \in CT_{11}$ be composed of the automaton $\{1, \ldots, 11\}$
  with local functions as follows:

  \begin{align*}
    f_1(x) &= x_1 \xor x_2 \xor x_3 \xor x_8 \xor x_9\\
    f_2(x) &= x_1 \xor x_2 \xor x_3 \xor x_4 \xor x_5\\
    f_3(x) &= x_1 \xor x_2 \xor x_3 \xor x_4 \xor x_5\\
    f_4(x) &= x_2 \xor x_4 \xor x_5\\
    f_5(x) &= x_2 \xor x_4 \xor x_5\\
    f_6(x) &= x_3 \xor x_6 \xor x_7\\
    f_7(x) &= x_3 \xor x_6 \xor x_7\\
    f_8(x) &= x_1 \xor x_8 \xor x_9\\
    f_9(x) &= x_1 \xor x_8 \xor x_9 \xor x_{10} \xor x_{11}\\
    f_{10}(x) &= x_9 \xor x_{10} \xor x_{11}\\
    f_{11}(x) &= x_9 \xor x_{10} \xor x_{11}.
  \end{align*}

  This solution is constructed from the initial network containing $\{1\}$,
  to which we add new automata two-by-two, in order; first
  $2$ and $3$ are added to $1$, then $4$ and $5$ added to $2$, and so on, until $10$ and $11$ are added to $9$. The interaction digraph
  of the resulting automata network is illustrated in Figure~\ref{fig-sol11}.

\end{example}

Note that the AN described in Example \ref{ex-sol11} is only one of all possible ANs in $CT_{11}$, as the vertex at which the new automata are added at each step may vary.

\def \xb{1.0}
\def \xo{-0.5}
\def \yb{0.866}
\def \d{2}

\begin{figure}[t!]
  \begin{center}
    \begin{tikzpicture} [->]

      \foreach \N\x\y\k in {
        A/0/0/1, B/0/1/2, C/-1/0/3, D/0/2/4, E/-1/1/5,
        F/-2/0/6, G/-2/-1/7, H/0/-1/8, I/1/0/9, J/2/1/10, K/2/0/11
      } {
        \node[state, minimum size = .7cm] (\N)
        at ({\x * \d * \xb + \y * \d * \xo}, {\y * \d * \yb}) {\k};
      }

      \def \bend{15}       
      \foreach \a\b\c in {A/B/C, B/D/E, C/F/G, A/H/I, I/J/K} {
        \draw
        (\a) edge[bend left= \bend] (\b)
        (\b) edge[bend left= \bend] (\a)
        (\b) edge[bend left= \bend] (\c)
        (\c) edge[bend left= \bend] (\b)
        (\a) edge[bend left= \bend] (\c)
        (\c) edge[bend left= \bend] (\a);
      }

      \def \angle{30}
      \def \looseness{7}
      \foreach \N\A in {
        A/50, B/50, C/110, D/90, E/210,
        F/150, G/270, H/270, I/110, J/90, K/330
      } {
        \draw (\N) to [out={\A - \angle/2}, in={\A + \angle/2}, looseness=\looseness] (\N);
      }

    \end{tikzpicture}
  \end{center}
  \caption{
    \label{fig-sol11}
    Interaction digraph of the network in $CT_{11}$ detailed in
    Example~\ref{ex-sol11}. Each local function is the XOR function of its
    influencing neighbors.
  }
\end{figure}
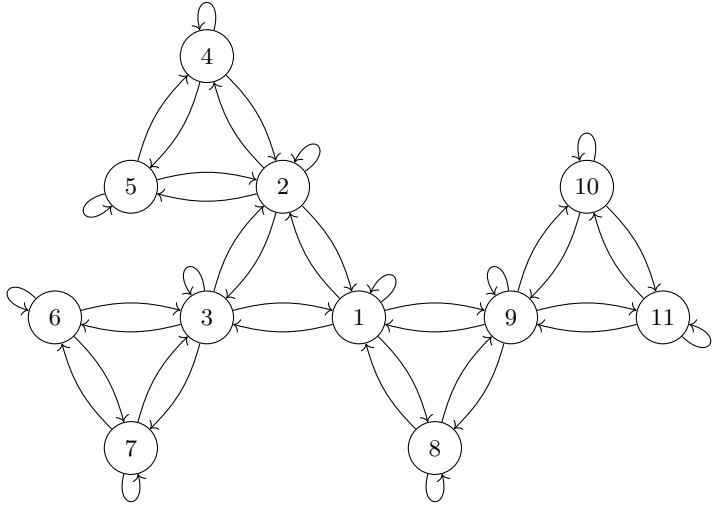

Some preliminary remarks and properties of this solution are that every automaton
is \emph{auto-regulated}, in the sense that, for any $F \in CT$ and $u$ an automaton of $F$, $(u, u)$ is an edge in the interaction digraph, and that every automaton $u$ is influenced by an odd number of automata in $F$. This can be seen by construction. First,
it is true for $F \in CT_1$; second, our construction of $CT_{n+2}$ involves two operations: 
it adds two new automata $n+1$ and $n+2$, both with $3$ interactions,
and it adds both of them as new influences to an automaton $i$, keeping the
amount of influences an odd number. Finally, we note that all influences are both ways;
if $(u, v)$ is an edge in the interaction graph, so is $(v, u)$.

Another set of remarks concern the overall structure of the graph. Our method
of constructing the solutions generates bigger solutions from smaller ones.
It can be seen as growing a solution from $CT_1$; in this process of growth,
it is important to remark that the distance between two automata is never
reduced; that is because while new automata are added at each step of the
construction, these automata are only connected to one automaton of the 
graph and never create any new smaller path within that graph.
As such, when considering any clique $(u, v, w)$,
the rest of the network can be partitioned in sets $V_u, V_v, V_w$ depending
on closest proximity to each automaton in $(u, v, w)$. We remark that there
is no possible ambiguity when doing this decomposition, and that any path
connecting $u$ to any automaton in $V_v$ or $V_w$
must go through $v$ or $w$,
respectively.

\subsection{Solution to $k$-parity}
\label{subsec:$k$-parity}

\begin{theorem}
\label{th-parity}
  All $F \in CT$ decide parity.
\end{theorem}
\begin{proof}
  Let $F \in CT$ be of size $n$. For $u \in \{1, \ldots, n\}$ and $r \in \N$,
  we denote by 
  $N_r(u)$ the neighbourhood of radius $r$ of $u$. That is, $N_0(u) = \{u\}$,
  and $N_{r + 1}(u)$ is defined as the union of $N_{r}(u)$ with the set of
  automata that influence the automata in $N_{r}(u)$.

  We claim the following: from any starting
  configuration $x \in \B^n$, for any number of steps $r$ and for any
  automaton $i \in \{1, \ldots, n\}$,
  $F^r(x)_i = \bigxor x|_{N_r(i)}$. That is, after $r$ steps, the value of
  every automaton of the network is the sum modulo 2 of their neighbourhood
  of radius $r$ over the initial configuration.

  This is evidently true for $r = 0$, as
  $F^0(x)_i = x_i$ $= \bigxor x|_{\{i\}} = \bigxor x|_{N_r(i)}$.
  We shall now assume the claim to be true for some $r$, and need to prove it
  for $r + 1$.

  By hypothesis, every value in $F^r(x)$ is the
  sum modulo 2 of a neighbourhood of radius $r$ in $x$. Let
  $i \in \{1, \ldots, n\}$ be an automaton, and $N_1(i)$ its immediate
  neighbourhood. When updating $i$, we take the sum modulo 2 of the neighbourhoods of
  size $r$ of every automaton in $N_1(i)$. We will now show that this XOR
  over the neighbourhoods of radius $r$ merge to form a sum modulo 2
  over the neighbourhoods of radius $r + 1$.


  First, let us remark that any automaton at distance $r - 1$ or less
  from $i$ is at distance $r$ from any node in $N_1(i)$; that
  is because every automaton
  in $N_1(i)$ is by definition at distance (at most) $1$ of $i$. Thus,
  when updating $i$, every automaton at distance $r - 1$ is added as
  many times as there are automata in $N_1(i)$, which is always an odd
  number.

  Second, any automaton at distance exactly $r$ from $i$ is also at distance
  $r$ from an odd number of automata in $N_1(i)$. If
  $|N_1(i)| = 3$, then it is at distance $r$ from all three automata in
  $N_1(i)$. If $|N_1(i)| > 3$, then it depends in which direction this
  automaton is; it will be at distance $r - 1$ of an even number of automata
  in $N_1(i)$, exactly at distance $r$ from $i$, and at distance $r + 1$
  from any other neighbour in $N_1(i)$. This is true because of the way the
  graph has been generated; the distances between two given automata is
  never reduced.
  Overall, it means that the values
  of the automaton at distance $r$ of $i$ are added an odd number of times
  in the update of $i$.

  Finally, any automaton at distance $r + 1$ from $i$ must be at distance
  $r$ of exactly one automata in $N_1(i)$, and it gets added exactly once
  in the update of $i$.
  
  Gathering all this, the update of $i$ after $r + 1$ is an addition modulo 2
  of the values in the initial configuration $x$
  of all of the automata at distance $r + 1$ from $i$, each added
  an odd number of times. This cancels out to simply be $\bigxor N_{r + 1}(i)$.

  Now that our claim is proven, simply consider that after a number of steps
  greater than the diameter of the graph, every automaton has converged to
  the sum modulo 2 of the initial configuration,
  which is both uniform and with every node in the state of the parity of the initial configuration. \qed
\end{proof}

\def \xb{1.0}
\def \xo{-0.5}
\def \yb{0.866}
\def \d{2}

\begin{figure}[t!]
  \begin{center}
    \begin{tikzpicture} [->]

      \foreach \N\x\y in {
        A/0/0, B/0/1, C/-1/0, D/1/0, E/0/-1,
        F/0/2, G/1/2, H/-2/0, I/-2/-1, J/-1/-2,
        K/0/-2, L/2/0, M/2/1
      } {
        \node[state, minimum size = .2cm] (\N)
        at ({\x * \d * \xb + \y * \d * \xo}, {\y * \d * \yb}) {};
      }

      \foreach \N\x\y in {
        FA/-1/2, FB/0/3, GA/2/3, GB/2/2,
        HA/-2/1, HB/-3/0, IA/-3/-2, IB/-2/-2,
        JA/-2/-3, JB/-2/-2, KA/0/-3, KB/1/-2,
        LA/2/-1, LB/3/0, MA/2/2, MB/3/2
      } {
        \node (\N) at ({\x * \d * \xb + \y * \d * \xo}, {\y * \d * \yb}) {};
      }

      \foreach \N in {F,G,H,I,J,K,L,M} {
        \foreach \M in {A,B} {
          \node (\N\M\M) at ($(\N)!0.75!(\N\M)$) {};
          \draw[densely dotted, -, thick]
          (\N) edge (\N\M\M);
        }
      }

      \def \bend{15}       
      \foreach \a\b\c in {A/B/C, A/D/E, B/F/G, C/H/I, E/J/K, D/L/M} {
        \draw
        (\a) edge[bend left= \bend] (\b)
        (\b) edge[bend left= \bend] (\a)
        (\b) edge[bend left= \bend] (\c)
        (\c) edge[bend left= \bend] (\b)
        (\a) edge[bend left= \bend] (\c)
        (\c) edge[bend left= \bend] (\a);
      }

      \def \angle{40}
      \def \looseness{7}
      \foreach \N\A\C in {
        A/60/5, B/0/3, C/300/3, D/120/3, E/180/3,
        F/240/1, G/300/1, H/60/1, I/0/1, J/120/1,
        K/60/1, L/240/1, M/180/1
      } {
        \draw (\N) to [out={\A + \angle/2}, in={\A - \angle/2}, looseness=\looseness] (\N);
        \path (\N) -- +(180+\A:0.35cm) node {\scriptsize $\times \C$};
      }

      \foreach \N in {A, B, C, D, E} {
        \draw[dashed] (\N) circle (1.25*\d);
      }
    \end{tikzpicture}
  \end{center}
  \caption{
    \label{fig-proof}
    Illustration of the recursion argument in the proof of Theorem~\ref{th-parity},
    after one iteration of the network.
    Labels such as $\times 3$ indicate that this particular automaton is to be counted
    $3$ times in the next evaluation of the center automaton, which itself
    is labelled by $\times 5$.
    By hypothesis, every automaton in the
    network has as value the sum modulo $2$ of their neighbourhood of radius $1$,
    illustrated here as dashed circles. In the next update, the center cell
    will thus take the sum modulo $2$ of all these neighbourhoods. As each
    automaton in the neighborhood of this sum is counted in an odd number of
    neighbourhoods, this overall sum turns out to be the sum modulo $2$ over the neighbourhood
    of radius $2$ of the center automaton over the initial configuration. The
    conclusion of this argument is that after enough updates, every automaton will
    be evaluated as the sum modulo $2$ of a neighbourhood that spans the entire
    network, proving the result.
  }
\end{figure}
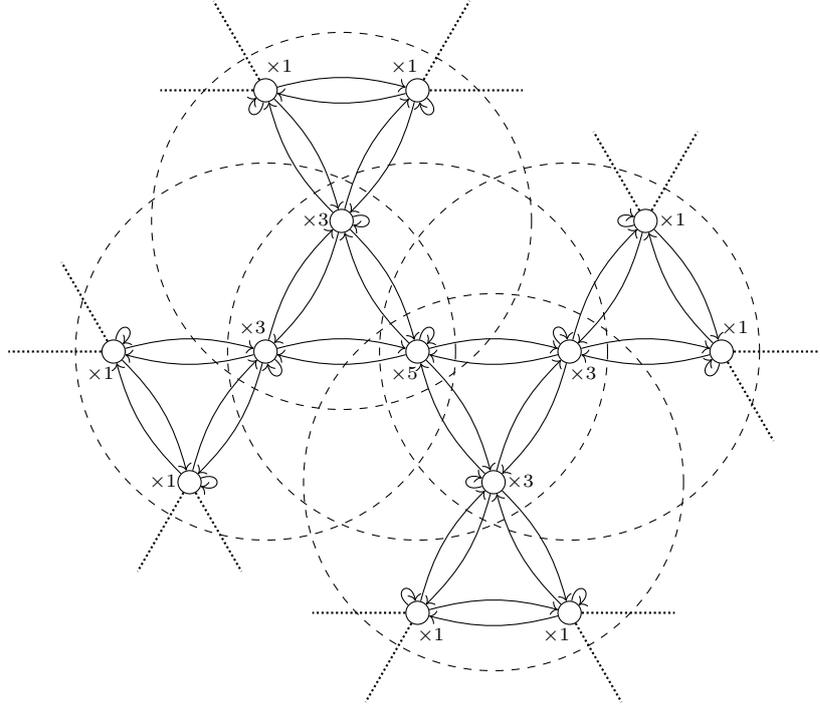

Figure~\ref{fig-proof} illustrates the main argument of the proof of Theorem~\ref{th-parity}.
From this proof we can directly extract the convergence time of the network,
which is always less or equal than the diameter of its interaction graph.
This is clearly optimal over these particular interaction graphs, as
any amount of time shorter than the diameter prevents information from getting
from one end to the other. Further arguments over the convergence time of the
solutions can be found in Section~\ref{sec-con}.
The temporal evolution of a solution of size 289 is illustrated in
Figure~\ref{fig-ca-parity}.

\begin{figure}
  \begin{center}
    \includegraphics{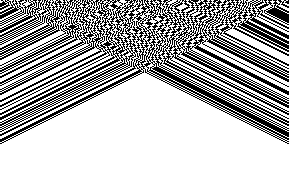}
  \end{center}
  \caption{
    \label{fig-ca-parity}
    Temporal evolution of a parity solution of size 289, with convergence time 144.
    This image is obtained from a solution in $CT_{289}$ which is constructed
    as a single chain of cliques, and to flatten its graphical representation
    into a line. Each line in the image represents a configuration of the
    network, where white means $1$ and black means $0$. Time goes down.
  }
\end{figure}

Our second result is that any deciding Boolean network that is negation
invariant provides a simple solution to the synchronisation problem.
A network is \emph{negation invariant} if $\neg F(x) = F( \neg x )$, for
all $x \in \B^n$. Notice that both of these conditions are always true in the
case of a parity solution built on the XOR function.

\subsection{Solution to $k$-synchronisation}
\label{subsec:k-synch}

\begin{theorem}
  If $F$ is a deciding automata network such that $\neg F(x) = F( \neg x )$,
  then $F' = \neg \circ F$ is a solution to the synchronisation problem.
\end{theorem}
\begin{proof}
  Since $\neg F(x) = F( \neg x )$, we have that 
 \begin{center}
     $F'( F' (x) ) = \neg F( \neg F ( x ) ) = \neg \neg F( F(x)) = F(F(x))$.
  \end{center}
  As $F$ decides, $F \circ F$ also decides, which means that
  $F' \circ F'$ decides, which means that it always converges.
  Note that on $0^n$ or $1^n$, $F'$ is simply the
  negation, and thus it synchronises. \qed
\end{proof}

This theorem allows us to adapt our previous $CT$ solution to the
synchronisation problem simply by taking any $F \in CT$ and adding a
negation on all local functions, providing solutions for any odd value of $n$.
Solutions for even $n$ can easily be provided by adding an extra automaton
that takes the negation of any other automaton, at the price of breaking
the strong connectivity of the network.
The temporal evolution of a solution of size 289 is illustrated in
Figure~\ref{fig-ca-synchronisation}.


\begin{figure}
  \begin{center}
    \includegraphics{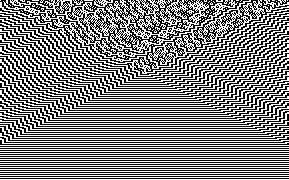}
  \end{center}
  \caption{
    \label{fig-ca-synchronisation}
    Temporal evolution of a synchronisation solution of size 289, with convergence time 144.
  }
\end{figure}

\section{Generalisation to arbitrary alphabet sizes}
\label{sec-lar}

We now prove that our construction generalises well to any finite alphabet
$\Sigma$. Let $k=|\Sigma|$, we call
$CT^k$ the families of solution that solve the $k$-parity problem.
This generalisation rests on the fact that the XOR function is also the 
sum modulo $2$; for an alphabet of size $k$, we will instead use the sum
modulo $k$ as the local function.
Where we added $2$ automata for each new step of
the construction of $CT^2_{n + 2}$, in general we will add $k$
automata in the construction of $CT^k_{n + k}$, while $CT^k_1$ stays
the identity network of size $1$ for all $k$.
The interaction graph of a such a solution over the ternary alphabet
is illustrated in Figure~\ref{fig-sol16}.

\def \xb{1.0}
\def \xo{-0.5}
\def \yb{0.5}
\def \d{2.5}

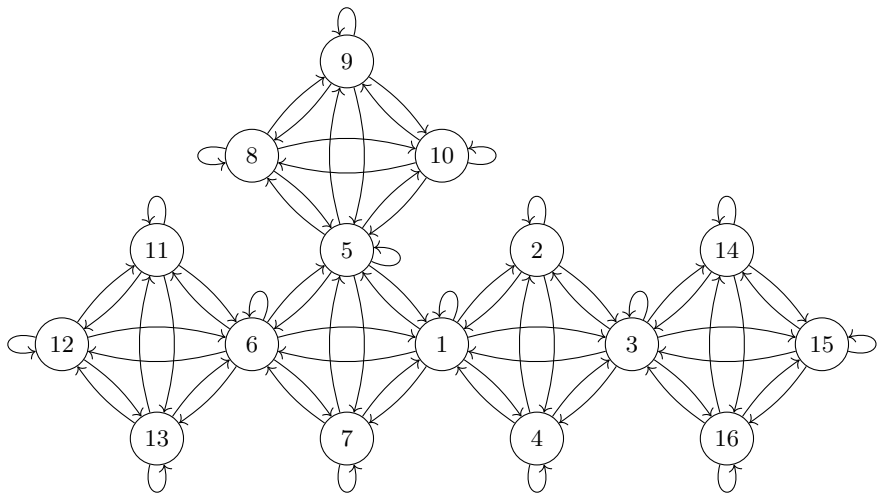
\begin{figure}[t!]
  \begin{center}
    \begin{tikzpicture} [->]

      \foreach \N\x\y\k in {
        A/3/4/9, B/2/3/8, C/3/3/10, D/1/2/11,
        E/2/2/5, F/3/2/2, G/4/2/14, H/0/1/12,
        I/1/1/6, J/2/1/1, K/3/1/3, L/4/1/15,
        M/0/0/13, N/1/0/7, O/2/0/4, P/3/0/16
      } {
        \node[state, minimum size = .7cm] (\N)
        at ({\x * \d * \xb + \y * \d * \xo}, {\y * \d * \yb}) {\k};
      }

      \def \innerbend{15}
      \def \outerbend{10}
      \foreach \a\b\c\d in {A/C/E/B, D/I/M/H, E/J/N/I, F/K/O/J, G/L/P/K} {
        \foreach \n\m in {\a/\c, \b/\d} {
          \draw
          (\n) edge[bend left= \innerbend] (\m)
          (\m) edge[bend left= \innerbend] (\n);
        }
        \foreach \n\m in {\a/\b, \b/\c, \c/\d, \d/\a} {
          \draw
          (\n) edge[bend left= \outerbend] (\m)
          (\m) edge[bend left= \outerbend] (\n);
        }
      }

      \def \angle{30}
      \def \looseness{7}
      \foreach \N\A in {
        A/90, B/180, C/0, D/90,
        E/350, F/90, G/90, H/180,
        I/80, J/80, K/80, L/0,
        M/270, N/270, O/270, P/270
      } {
        \draw (\N) to [out={\A - \angle/2}, in={\A + \angle/2}, looseness=\looseness] (\N);
      }

    \end{tikzpicture}
  \end{center}
  \caption{
    \label{fig-sol16}
    Interaction digraph of a solution in $CT^3_{16}$.
    Each local function is the sum modulo 3 of its
    influencing neighbours. The automata are numbered in one of the
    possible orders of addition that construct this solution.
  }
\end{figure}

In this generalisation, deciding means always converging to any fixed point
$s^n$ for some $s \in \Sigma$. The generalisation of the parity problem is
now the $k$-parity problem; to solve it, a network with alphabet $\Sigma$
must always converge to the fixed point with symbol the sum modulo $k$ of
the values of the initial configuration.


The generalisation of the synchronisation problem is the $k$-synchronisation
problem, where instead of oscillating between
$0^n$ and $1^n$, a $k$-synchronising network over an alphabet of size $k$ must
cycle through the $k$ previously described fixed points. While our construction
cycles through these fixed points in increasing order,
the order of this
cycle is unimportant, as a solution with a different order can always be
converted to a solution in any order.

\begin{theorem}
  All $F \in CT^k$ solve the $k$-parity problem, i.e., the sum modulo $k$ decision problem.
\end{theorem}

\begin{proof}
  This proof follows the same arguments as the binary case; given $x$ a starting
  configuration, initially every automaton computes the sum modulo $k$ of
  its neighbourhood of radius $0$.
  
  Let us suppose that for values up to some $r$,
  all automata compute the sum modulo $k$ of the initial configuration
  over their neighbourhood of radius $r$
  after $r$ steps. By the construction of the interaction digraph,
  the next update of all automata
  will result in a sum modulo $k$
  of the value in the initial configuration of their neighbours of distance
  up to $r + 1$,
  each neighbour being counted $m \times k + 1$ times, for $m$
  some positive integer.

  As a result, after $r$ steps, every automaton computes the sum modulo $k$ of
  the values in the initial configuration of their neighbourhood of radius $r$. After
  $r = n$ steps, the network has solved the $k$-parity problem.
\end{proof}

The temporal evolution of a solution to the 3-parity problem over a configuration
of size 289 is illustrated in Figure~\ref{fig-ca-3arity}.

\begin{figure}
  \begin{center}
    \includegraphics{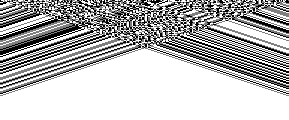}
  \end{center}
  \caption{
    \label{fig-ca-3arity}
    Temporal evolution of a 3-parity solution of size 289, with convergence time 96, flowing downwards. White means $2$, gray means $1$ and black means $0$.
  }
\end{figure}

Similarly, our argument that transformed the parity solution into a
synchronisation solution also works for any finite alphabet; but instead of
the invariance to be over the negation, we now require it to be over the
$+1$ operator.

\begin{theorem}
  If $F$ decides and is $+1$ invariant, then $F' = (+1) \circ F$ is a solution
  to the $k$-synchronisation problem.
\end{theorem}
\begin{proof}
  This proof follows the same arguments as the binary case: the $k$ repetition of
  $F'$ is equivalent to the $k$ repetition of $F$ thanks to the invariance hypothesis.
  This means that since $F$ converges, $F^k$ also converges, and so $F'^k$
  converges to a uniform configuration; this also implies that $F'$ always
  converges to the cycles that go through every uniform fixed point in increasing order.
\end{proof}

The temporal evolution of a synchronisation solution over the ternary alphabet and
a configuration of size 289 is illustrated in Figure~\ref{fig-ca-3sync}.

\begin{figure}
  \begin{center}
    \includegraphics{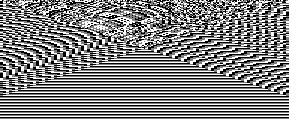}
  \end{center}
  \caption{
    \label{fig-ca-3sync}
    Temporal evolution of a 3-synchronisation solution of size 289, with convergence time 96, flowing downwards. White means $2$, gray means $1$ and black means $0$.
  }
\end{figure}

Whereas our the binary solution can be described as a collection
fully connected cliques of size $3$ lightly connected to each other,
the $k$-parity solution can be described
as a collection of fully connected cliques of size $k + 1$ lightly connected to each other.


\section{Convergence times}
\label{sec-con}

\subsection{Direct convergence}

As a final analysis point, we present some facts about the convergence times
of the presented solutions.

Central to the convergence proof is the fact that each automaton computes
the sum modulo $k$ of its neighbourhood of radius $r$ after $r$ steps; this
means that every automaton will have converged as soon as $r$ is large enough
to include the entire graph. This value is exactly the diameter of the graph.

The interaction digraph of any solution in $CT^k$ could be described as a
collection of cliques of size $k + 1$ connected to
each other through their nodes such as they form a tree; for any $n$ and $k$,
the solution in $CT^k_n$ with the maximal diameter is the solution composed
of cliques that together form a line. Let us call $C$ the number
of such components and $D$ the diameter.
In the case of a line, the diameter
is always the number of cliques: $D = C$.
Every connected component is of size $k + 1$.
By construction
every connected component shares a node with the next, except for the last one;
this means that overall, the network has $n = (C - 1) \times k + (k + 1)$ automata.
We can infer the following:

\begin{align*}
  n &= (C - 1) \times k + (k + 1) = \\
  &= (D - 1) \times k + k + 1 = \\
  &= D \times k + 1,\\
\end{align*}

\noindent which implies that $D = \frac{n - 1}{k}$. In the binary case this means that
the convergence times tends to $\frac{n}{2}$ as $n$ grows larger.
Notice also that the convergence time decreases
as $k$ grows larger for a given $n$; this is explained by the fact that
a larger $k$ implies larger cliques which means that
the network is more connected and has a smaller diameter relative to its size.
This fact can be exploited to increase convergence time for a given $k$.

\subsection{Trading component size for convergence time}

For a given $k$, we can decrease convergence time as close to $0$ as desired
by first constructing a solution for an alphabet of size a multiple of $k$,
and then projecting that solution unto the alphabet of size $k$. By construction,
this solution will converge faster, and it can be checked that it stills solves
the sum modulo $k$ problem. For $m \in \N$, we denote
$\pi_{mk,k} : \{0, \ldots, mk - 1\} \to \{0, \ldots, k - 1\}$ the projection which maps
any $v$ to the reminder of the Euclidean division of $v$ by $k$.


\begin{theorem}
  If $F \in CT^{mk}$, then $\pi_{mk, k} \circ F \in CT^k$.
\end{theorem}
\begin{proof}
  Since after $r$ steps every automaton in $F$ computes the
  sum modulo $mk$ of its neighbourhood of radius $r$, it also follows
  that every automaton in $\pi_{mk,k} \circ F$ computes the sum modulo
  $k$ of its neighbourhood of radius $r$ in the same time.
\end{proof}

The temporal evolution of a solution to the parity problem that converges twice as
fast compared to the base solution
is illustrated in Figure~\ref{fig-ca-fastparity}.

Hence for a given network with $n$ automata we can cut the convergence time of the network by
a factor of $m$ by simply considering the solution for the alphabet of size
$mk$, and projecting down. The trade-off of this acceleration is that this new
solution needs to be constructed out of cliques of size $mk + 1$,
which means that each automaton has at least $mk + 1$ neighbours.
In other terms, the number of edges in the interaction graph grows as a polynomial of $k$.

\begin{figure}
  \begin{center}
    \includegraphics{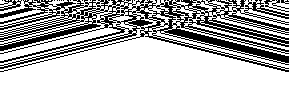}
  \end{center}
  \caption{
    \label{fig-ca-fastparity}
    Temporal evolution of a faster parity solution of size 289, with convergence time 72, flowing downwards.
    This solution is the combination of the solution $CT^4_{289}$ together with
    a projection to the binary alphabet. It is guaranteed to converge twice
    as fast compared to any solution in $CT^2_{289}$.
  }
\end{figure}

This reasoning reaches its natural conclusion by simply taking $k = n - 1$
and then projecting down, in which case the interaction digraph simply
becomes the complete graph of size $n$, in which case computing the
parity of the initial configuration only takes one step.

All of these observations also apply to our solutions to the $k$-synchronisation
problem. The temporal evolution of such a solution in the binary synchronisation case
projected down from the alphabet of size $4$ is illustrated in
Figure~\ref{fig-ca-fastsync}. 


\begin{figure}
  \begin{center}
    \includegraphics{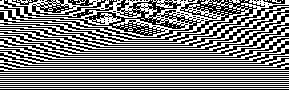}
  \end{center}
  \caption{
    \label{fig-ca-fastsync}
    Temporal evolution of a faster synchronisation solution of size 289,
    with convergence time 72, flowing downwards. It has been obtained from a solution over
    the alphabet $\{0, 1, 2, 3\}$ combined with a projection to the
    binary alphabet.
  }
\end{figure}

\section{Concluding remarks}

In this paper we presented solutions to the parity and
synchronisation problems that are generalisable to any finite alphabet,
using the flexible geometry of automata networks.
The convergence time of these solutions is easy to compute and is always
the diameter of their interaction graph. Further results are provided
to accelerate this convergence time by projecting down a solution over
an alphabet the size of which is a multiple of the desired solution's alphabet.

Figures \ref{fig-ca-parity}, \ref{fig-ca-synchronisation}, \ref{fig-ca-3arity}, \ref{fig-ca-3sync}, \ref{fig-ca-fastparity} and \ref{fig-ca-fastsync} all use configurations of size 289 specifically
because 289 is one plus a multiple of $3$ and $4$, which allows these configurations
to work in the context of their specific problem, while keeping a coherent pixel resolution
throughout the illustrations of the paper.

The geometry of our solutions is more complex than
the geometry of
standard cellular automata lattices, yet they conserve a simple, local
and repeatable form. We believe that this showcases interesting
possibilities for designing new solutions to convergence problems
in the realm of cellular automata.
We believe other problems such as the well-known  density classification task~\cite{density}
could be approached this way.

\section*{Acknowledgements}

P.P.B. thanks the Brazilian agencies CNPq (Conselho Nacional de Desenvolvimento Cient\'{i}fico e Tecnol\'{o}gico) for the research grant PQ 303356/2022-7, and CAPES (Coordena\c{c}\~{a}o de Aperfei\c{c}oamento de Pessoal de N\'{i}vel Superior) for Mackenzie-PrInt research grant no.\ 88887.310281/2018-00. P.P.B. and E.R. jointly thank CAPES for the research grant STIC-AmSud no.\ 88881.694458/2022-01, and P.P. thanks CAPES for the postdoc grant no.\ 88887.833212/2023-00.

\bibliographystyle{plain}
{\small{\bibliography{bib}}}

\end{document}